\documentclass[journal]{ieeetran}
\usepackage{amsmath,amsthm,amssymb,paralist,subfigure,cite,graphicx,color}
\usepackage{algorithm2e}
\newtheorem{lem}{Lemma}
\newtheorem{thm}{Theorem}
\newtheorem{defn}{Definition}

\newtheorem{cor}{Corollary}
\newtheorem{prob}{Problem}

\def\mb{\mathbf}

\def\mc{\mathcal}

\DeclareMathOperator*{\argmin}{argmin}

\IEEEoverridecommandlockouts
\begin{document}
\title{Minimal Sufficient Conditions for Structural Observability/Controllability of Composite Networks via Kronecker Product}
\author{Mohammadreza Doostmohammadian and Usman A. Khan, \textit{Senior Member}, \textit{IEEE}
\thanks{Mohammadreza Doostmohammadian is with the Mechanical Engineering Department at Semnan University, Semnan, Iran, email: \texttt{doost@semnan.ac.ir}. Usman A. Khan is with the Electrical and Computer Engineering Department at Tufts University, Medford, USA, email: \texttt{khan@ece.tufts.edu}. The work of UAK has been partially supported by NSF under awards \#1350264, \#1903972, and \#1935555.}}
\maketitle

\begin{abstract}
	In this paper, we consider composite networks formed from the Kronecker product of smaller networks. We find the observability and controllability properties of the product network from those of its constituent smaller networks. The overall network is modeled as a \textit{Linear-Structure-Invariant (LSI)} dynamical system where the underlying matrices have a fixed zero/non-zero structure but the non-zero elements are potentially time-varying. This approach allows to model the system parameters as free variables whose values may only be known within a certain tolerance. We particularly look for minimal sufficient conditions\footnote{We emphasize that a minimal sufficient condition is not necessarily a necessary and sufficient condition. In fact, it implies that among all sufficient conditions that may result in an event, this condition is the least conservative but usually is not necessary; see~\cite{mackie1965causes} for details.} on the observability and controllability of the composite network, which have a direct application in distributed estimation and in the design of networked control systems. The methodology in this paper is based on the structured systems analysis and graph theory, and therefore, the results are generic, i.e., they apply to almost all non-zero choices of free parameters. We show the controllability/observability results for composite product networks resulting from full structural-rank systems and self-damped networks. We provide an illustrative example of estimation based on Kalman filtering over a composite network to verify our results.
	
\end{abstract}

\section{Introduction} \label{sec_intro}
\IEEEPARstart{C}{ontrollability}  and observability of networked systems arise in situations where a group of interconnected devices are influenced or observed by an external entity; examples range from classical applications in robotic systems, multi-agent networks, and sensor networks, to more recent emerging areas that include social networks, Internet-of-Things (IoT) and Cyber-Physical Systems (CPS)~\cite{gao2012networks,forti2017distributed,fawzi2014secure,miorandi2012internet,gorod2008system,sadreazami2017distributed,chen2018internet,guan2017distributed}. In all of these applications, there is a layered network that connects individuals and objects leading to a significant interest in composite networks~\cite{chapman2014controllability,carvalho2017composability,rech1991structural,li1996g,davison1977connectability}. The main theme in these works is to investigate the observability and controllability of large-scale networks resulting from the \textit{product} of smaller networks. Two main network products that arise in real applications are \textit{Cartesian product} and \textit{Kronecker product}\footnote{Kronecker product of two graphs is also referred to as tensor product or direct product in the literature~\cite{hammack2011handbook}.}. We refer interested readers to~\cite{hammack2011handbook} for a better understanding of different graph products. An example of Cartesian product vs. Kronecker product is given in Fig.~\ref{fig_graphproduct}. As it can be seen from the figure, Kronecker product, as compared to Cartesian product, typically results in more complicated composite graphs that may even be bipartite in terms of connectivity (see the example in Fig.~\ref{fig_graphproduct}). Different applications of Kronecker composite networks may be found in~\cite{leskovec2008dynamics}. We are particularly interested in the observability and controllability of Kronecker composite networks, whereas in contrast most of the related literature studies composite networks via Cartesian products~\cite{chapman2014controllability,carvalho2017composability,rech1991structural,li1996g,davison1977connectability}. The Kronecker composite networks find direct applications in networked control system~\cite{hu2008stability,ling2003soft,lee2015stability}, distributed Fault Detection and Isolation (FDI) ~\cite{davoodi2013distributed}, distributed detection~\cite{cattivelli2011distributed}, and distributed estimation over sensor networks~\cite{doostmohammadian2018structural,jstsp14,nuno-suff.ness,park2017design, jstsp,usman_cdc:10,acc13}. For example, in distributed estimation, the overall distributed system can be considered as Kronecker product of the system digraph and the sensor/estimator network.
\begin{figure}[!t]
	\centering
	{\includegraphics[width=3.3in]{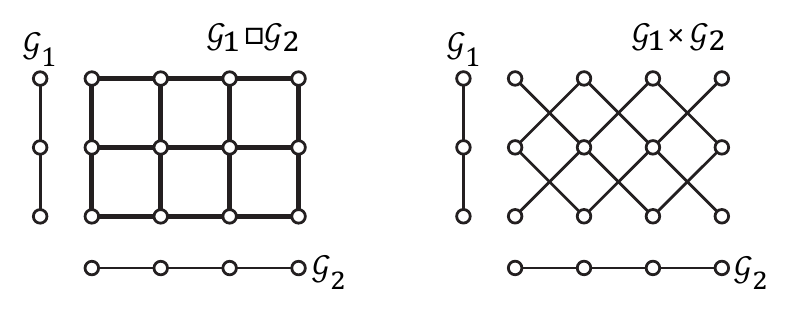}}
	\caption{ (Left) The Cartesian product of two line graphs~$\mc{G}_1$ and~$\mc{G}_2$. (Right) The Kronecker product of the same graphs.}
	\label{fig_graphproduct}
\end{figure}

In this paper, we discuss the \textit{structural} controllability and observability of the composite networks representing dynamical systems. We assume the underlying dynamical system to be LSI which is prevalent, among others, in social networks~\cite{jstsp14,friedkin2006structural,ISJ,French}. In LSI systems, the system structure (representing the network) is fixed while the system parameters (and consequently the link weights in the network) may vary over time. In other words, the system dynamics generated from interaction of system parameters is time-invariant while the values of the system parameters may change as free variables. Such LSI systems may also arise in linearization of nonlinear systems~\cite{nonlin}. It is known that many properties of such systems are generic, implying that they do not depend on the exact numerical values of system parameters, rather the system structure~\cite{woude:03}. Controllability and observability are examples of such properties. The structural controllability is primarily  introduced in~\cite{lin} and further developed in~\cite{pequito2015framework,monshizadeh2014zero,burgarth2013zero,mousavi2017structural,acc13_mesbahi,mousavi2019strong,trefois2015zero,menara2018structural}. In this direction, the concept of structural controllability holds for almost all choices of network link weights~\cite{pequito2015framework}, while the concept of \textit{strong} structural controllability holds for all choices of the system parameters and network link weights~\cite{monshizadeh2014zero,burgarth2013zero,mousavi2017structural,acc13_mesbahi,mousavi2019strong}. In case of strong structural controllability, for example, see the concept of qualitative class of matrices defined in~\cite{monshizadeh2014zero}. The analysis of strong structural controllability is mainly through two key notions: the zero forcing sets~\cite{monshizadeh2014zero,burgarth2013zero,mousavi2017structural,mousavi2019strong}, and the
constrained bipartite matching\footnote{The constrained bipartite matching is defined over bipartite representation of the network. In the bipartite representation, define a~$t$-matching as a set of~$t$ links such that no two of them share a node. A~$t$-matching is called constrained if there is no other~$t$-matching with the same matched nodes. For more information refer to~\cite{trefois2015zero}.}~\cite{acc13_mesbahi,trefois2015zero}. In~\cite{monshizadeh2014zero} a one-to-one correspondence between the set
of driver nodes (the nodes rendering the network controllable) of directed graphs and the zero forcing sets is established. One basic assumption in structural controllability
is the independency of all free variables. In this direction, the controllability of undirected networks with symmetric adjacency matrices is studied in~\cite{mousavi2017structural,menara2018structural} as a special case of having dependency in system parameters. The results are particularly of interest as the undirected link in a dynamic network may represent a feedback loop. In~\cite{mousavi2019strong}, the authors explore the strong structural controllability under network perturbations, where they characterize the addition or removal of maximal set of links for which the strong structural controllability is still preserved. In another line of research, Ref.~\cite{cowan2012nodal} claims that nodal dynamics, not degree distributions, determine the structural controllability of complex networks. The authors define power dominating sets (PDS) and structural control network as key concepts concerning network controllability. Similarly, the work in~\cite{zhou2017minimal} shows  that the minimal number of the inputs/outputs of a system is equal to the maximum geometric multiplicity of its state transition matrix, which is consistent with~\cite{cowan2012nodal}. The results in~\cite{zhou2017minimal} are in sharp contrast to minimum controllability/observability problem in some literature, which claim that the sparsest input/output matrix subject to system controllability/observability is NP-hard and even impossible to be approximated within a multiplicative factor~\cite{zhou2017minimal}.

Considering the Kronecker composite network~$\mc{G}_1 \times \mc{G}_2$, an important problem is to design the network~$\mc{G}_1$ representing the sensor network or the control actuation network, based on the properties of the network~$\mc{G}_2$ representing the underlying dynamical system. In this direction, what is missing from the literature, including~\cite{carvalho2017composability,chapman2014controllability} and references therein, is the effect of network S-rank (or structural-rank) on the composability properties. In this paper, we show that the controllability and observability properties of the Kronecker composite network to a great extent depend on the S-rank of the underlying networks and, further, their strong connectivity. We find the minimal  conditions to guarantee observability/controllability of the composite network and provide sufficient conditions on the structure of the network~$\mc{G}_1$ and the observation/input matrix~$\mc{H}_C$.
First, similar to~\cite{pequito2015framework}, we find the minimal conditions on the observability/controllability of the general constituent networks, and then extend the results to the product network using structured systems theory. The main results on observability/controllability of the product network	are based on the S-rank of the constituent networks and particularly the case of \textit{self-damped} constituent networks.
In particular, for proving our results we adopt the structural methodology in~\cite{rein_book}. To the best of our knowledge, no result in the literature is developed on the structural observability and controllability of the composite Kronecker product networks, and the few related papers study the case of Cartesian product network~\cite{carvalho2017composability,chapman2014controllability}. We again remind the reader that the above problem on observability/controllability of Kronecker product finds signal processing applications such as in distributed estimation/detection~\cite{doostmohammadian2018structural,jstsp14,nuno-suff.ness,park2017design, jstsp,usman_cdc:10} and control applications in networked control systems~\cite{hu2008stability,ling2003soft,lee2015stability}.

The rest of the paper is organized as follows. In Section~\ref{sec_prob}, we state the preliminaries on graphs and structured systems theory and formulate the problem. In Section~\ref{sec_cent}, we discuss the main results on minimal conditions for observability/controllability and extend these results to minimal sufficient observability/controllability of the Kronecker composite network in Section~\ref{sec_main}. In Section~\ref{sec_app}, we provide a representative application in distributed estimation. In Section~\ref{sec_sim}, we provide an illustrative example and simulations. Finally, Section~\ref{sec_con} states the concluding remarks.

\section{Preliminaries and Problem Formulation} \label{sec_prob}

\subsection{Problem Statement}
In this paper, we consider the composite networks based on Kronecker product, represented by~$\mc{G}_1 \times \mc{G}_2$. In the coming sections, we refer to~$\mc{G}_2$ as the \textit{replica} network with adjacency matrix~$A_2 \in \mathbb{R}^{n \times n}$ and ~$\mc{G}_1$ as the \textit{factor} network with adjacency matrix~$A_1 \in \mathbb{R}^{N \times N}$. The matrix~$A_1 \otimes A_2 \in \mathbb{R}^{nN \times nN}$ represents the adjacency matrix of the composite network (or \textit{network-of-networks} as referred by~\cite{chapman2014controllability}), where~$\otimes$ is  the Kronecker matrix product. We refer interested readers to~\cite{hammack2011handbook,leskovec2008dynamics} for discussion on properties of Kronecker product networks. The replica network~$\mc{G}_2$ may represent a dynamical system and the factor network~$\mc{G}_1$ may represent a sensor network or a network of control actuation devices. The structured matrices~$\mc{A}_1 \in \{0,1\}^{N \times N}$ and~$\mc{A}_2 \in \{0,1\}^{n \times n}$ respectively represent the structure (0-1 pattern) of the matrices~$A_1$ and~$A_2$. The size of the networks~$\mc{G}_1$ and~$\mc{G}_2$ are respectively represented by~$N$ and~$n$. The matrix~$H$ and~$H_C$ are, respectively, the measurement/input matrix of the replica network~$\mc{G}_2$ and of the composite network,~$\mc{G}_1 \times \mc{G}_2$. Their structured representation (0-1 pattern) is denoted by~$\mc{H}$ and~$\mc{H}_C$.
In this paper, we define a network~$\mc{G}=\{\mc{V},\mc{E}\}$ where~$\mc{V}$ and~$\mc{E}$ are the set of nodes and links respectively.~$(\mc{V}_i,\mc{V}_j)$ is in the set~$\mc{E}$ if there is a link from node~$\mc{V}_i$ to node~$\mc{V}_j$. Further, define the neighborhood set of a node~$\mc{V}_j$ as~$\mc{N}(\mc{V}_j) = \{\mc{V}_i|(\mc{V}_i,\mc{V}_j) \in \mc{E}\}$.
Note that we use the well-known definition for adjacency matrix. The entry~$A_{ij}$ of the adjacency matrix~$A$ is equal to the weight of the link~$(\mc{V}_i,\mc{V}_j)$. The diagonal entries of~$A$ represent the weights of self-links of the graph~$\mc{G}$.

The problem we study is that given the replica network~$\mc{G}_2$, what are the  sufficient conditions on the factor network~$\mc{G}_1$ and the measurement/input matrix~$\mc{H}_C$ to ensure the structural observability/controllability of the composite network,~$\mc{G}_1 \times \mc{G}_2$? Particularly, we are interested in minimum number of measurements/inputs over the composite network. Mathematically, we aim to find the solution for the following problem,
\begin{prob} \label{prob}
\begin{equation} \label{eq_prob}
 \begin{aligned}
 \displaystyle
 \argmin
 \limits_{\mc{A}_1,\mc{H}_C} ~~ & |\mc{H}_C|_0 \\
 \text{s.t.} ~~ & (\mc{A}_1 \otimes \mc{A}_2,\mc{H}_C)\mbox{-observability/controllability,}
 \end{aligned}
\end{equation}
\end{prob}
\noindent where~$|\cdot|_0$ is the standard~$0$-norm definition that counts the number of non-zero elements (free entries) in the matrix (or number of~$1$'s in the structured matrix). Indeed,~$|\mc{H}_C|_0$ represents the number of outputs/inputs of the composite network. We note here that the observability or controllability constraint on~$(\mc{A}_1 \otimes \mc{A}_2,\mc{H}_C)$ is \textit{structural}, i.e, we derive the results only from the 0-1 pattern (or the structure) of the matrices,~$H_C$ and~$A_1$, and not their exact numerical values. In other words, the results we derive are \textit{generic} and they are applicable to almost all\footnote{The term ``almost all" mathematically implies that the probability of randomly choosing the free variables for which the specific property does not hold is zero. In other words, the parameter values violating the property lie on an algebraic subspace with zero measure in Lebesgue sense. Therefore, the probability of randomly lying on this subspace is mathematically zero~\cite{woude:03}. This random approach is irrespective of the distribution type and holds for any random selection.} free entries (the entries which are not fixed zeros of the system matrix) in the system matrices as long as the structure is not violated. The parameter values for which the generic results do not hold lie on an algebraic variety of measure zero, see, e.g.,~\cite{woude:03} for details.
In this direction, the networks in this paper may represent  LSI dynamical systems, where the system matrix is fixed-structure with possibly time-varying  entries. In other words, the system dynamics is in the form~$\dot{x} = Ax$ where the structure of~$A$ is fixed while its entries may potentially change in time. The only concern is that the time-varying entries do not lie on a certain zero-measure algebraic subspace. It should be mentioned that in this paper we consider the constituent networks to be directed with general non-symmetric adjacency matrices, while controllability/observability of undirected networks is considered in~\cite{mousavi2017structural,menara2018structural}.

\textit{Assumptions:} In this paper, we assume the structure of the directed networks~$\mc{G}_1$ and~$\mc{G}_2$ to be known and time-invariant. Further, the entries of their adjacency matrices~${A}_1$ and~${A}_2$ (or the link weights in the networks) are assumed as free variables.

\subsection{Some Useful Lemmas and Definitions}
\begin{defn}
	Define the S-rank or structural rank as the maximum numerical rank of the network adjacency matrix over all possible non-zero entries (or all possible values for the link weights in the network)~\cite{harary,murota}. 
\end{defn}
In fact, it is known that this is the maximum rank that can be attained over almost all choices of free variables~\cite{woude:03}. From this point onwards in the paper, for simplicity, we refer the network as being full S-rank or S-rank-deficient.

\begin{lem} \label{lem_srank}
	A network is full S-rank if and only if it contains a disjoint family of cycles spanning all nodes.
\end{lem}
\begin{proof}
	The proof directly follows the definition of spanning cycle family~\cite{rein_book} and the proof of Theorem 1 in~\cite{harary}.
\end{proof}

\begin{lem}\label{lem_srankproduct}
	The following holds for S-rank of Kronecker product of adjacency matrices of two networks:
	\begin{eqnarray} \label{eq_kronrank}
	\mbox{S-rank}(A_1 \otimes A_2) \geq \mbox{S-rank}(A_1)\cdot \mbox{S-rank}(A_2)
	\end{eqnarray}
where the equality holds for \textit{almost all} choices of  matrices~$A_1$ and~$A_2$.
\end{lem}

\begin{proof}
	The above lemma directly results from the definition of S-rank as the maximum possible rank of the adjacency matrix. Note that the S-rank is a generic property~\cite{woude:03}. Following the fact that~$	\mbox{rank}(A_1 \otimes A_2) = \mbox{rank}(A_1)\cdot \mbox{rank}(A_2)$~\cite{schacke2004kronecker} and the genericity of the Kronecker product, one can conclude that the equality in~\eqref{eq_kronrank} holds for almost all choices of matrices~$A_1$ and~$A_2$. Further, the choices of entries of~$A_1$ and~$A_2$ for which~$\mbox{rank}(A_1 \otimes A_2) > \mbox{rank}(A_1)\cdot \mbox{rank}(A_2)$ lie an algebraic subspace of zero Lebesgue measure in~$Nn$-dimensional space.
\end{proof}

The following is a well-known lemma on (structural) observability/controllability of networks and is originally taken from~\cite{rein_book}.
\begin{lem} \label{lem_cent}
	A network is (structurally) observable/controllable if and only if the following conditions hold:
	
	(i) \textbf{output/input-connectivity}: every node in the network is a begin/end-node of a connected path to/from an output/input.
	
	(ii) \textbf{S-rank recovery}: there exist a disjoint family of cycles and output/input-connected paths spanning all nodes in the network.
\end{lem}
\begin{proof}
	The detailed proof for structural controllability is given in~\cite{rein_book}, and one can easily extend the proof to dual notion of structural observability.
\end{proof}
The following lemma gives a restatement of condition (ii) in Lemma~\ref{lem_cent}.
\begin{lem} \label{lem_srank2}
	The condition (ii) in Lemma~\ref{lem_cent} recovers the S-rank of the network. For example in case of observability, if ~$\mbox{S-rank}(\mc{A}_2)<n$ the condition (ii) on the (structured) observation matrix~$\mc{H}$  is equivalent to the following:
	\begin{equation}
	\mbox{S-rank}\left(
	\begin{array}{c}
	\mc{A}_2\\
	\mc{H}
	\end{array}
	\right) = n.
	\end{equation}
\end{lem}	
\begin{proof}
	The proof can be found in~\cite{davison1977connectability}.
\end{proof}	
In the line of above results on structured systems theory a very related concept is \textit{self-damped} network (system) defined as follows:

\begin{defn}
	Define the \textit{self-damped} network as a network having self-links (or self-cycles) at every node~\cite{acc13_mesbahi}.
\end{defn}

	
\section{ Minimal Conditions for~$(\mc{A}_2,\mc{H})$-Observability/Controllability } \label{sec_cent}
To solve the Problem~\ref{prob}, we first need to find the minimal outputs/inputs for observability/controllability of the replica network~$\mc{G}_2$. Mathematically, we solve the following problem in this section,
\begin{prob} \label{probA}
\begin{equation} \label{eq_probA}
\begin{aligned}
\displaystyle
\argmin
\limits_{\mc{H}} ~~ & |\mc{H}|_0 \\
\text{s.t.} ~~ & ( \mc{A}_2,\mc{H})\mbox{-observability/controllability.}
\end{aligned}
\end{equation}
\end{prob}
\noindent To solve this problem, we borrow some definitions and concepts from our previous works~\cite{jstsp,asilomar11}. Similar definitions are introduced in~\cite{pequito2015framework}. We define specific components in the network, which are involved in the observability/controllability of networks as follows.
\begin{defn}
	A Strongly Connected Component (SCC) is defined as a component in which every node is connected to every other node via a directed path. Define a Strongly-Connected (SC) network as a network in which there is a directed path from every node to every other node, i.e., the entire network makes one SCC. In a non-SC network, define a \textit{parent} SCC as a SCC having no outgoing links to nodes in any other SCC. Similarly, a \textit{child} SCC is defined as a SCC with no incoming links from nodes in other SCCs. Further,~$\mc{S}^p=\{\mc{S}^p_1,\mc{S}^p_2,\dots\}$ and~$\mc{S}^c=\{\mc{S}^c_1,\mc{S}^c_2,\dots\}$ respectively represent the set of parent SCCs and  child SCCs in the network, and the partial order\footnote{Among SCCs, the partial order is defined as having incoming links (or a sequence of links) from other SCCs, i.e.,~$\mc{S}_a \prec \mc{S}_b$ implies that there is (at least) a link/path from a node in~$\mc{S}_a$ to a node in~$\mc{S}_b$.} is defined as~$\mc{S}^c_i \prec \mc{S}^p_j$.
\end{defn}
The SCCs and their partial order are specifically related to the output/input-connectivity condition in Lemma~\ref{lem_cent} as stated in the following.
\begin{thm} \label{thm_scc}
	Let~$\mbox{S-rank}(\mc{A}_2)=n$. For observability/controll-ability of~$\mc{G}_2$, it is necessary and sufficient to measure/control one node in every parent/child SCC, i.e., for minimum observability~$|\mc{H}|_0=|\mc{S}^p|$ and for minimum controllability~$|\mc{H}|_0=|\mc{S}^c|$, where~$|\cdot|$ represents the cardinality of the set.
\end{thm}
\begin{proof}
	The proof follows from the conditions in Lemma~\ref{lem_cent}. The network being full S-rank implies that the condition (ii) in Lemma~\ref{lem_cent} is satisfied. For condition (i) every node must satisfy output/input-connectivity. We state the proof for observability and it can be easily extended to controllability.
	
	\textbf{Sufficiency}: Consider one node measurement  in every~parent~SCC. Following the definition of SCC, all nodes in the same parent SCC have a connected path to the outputs. Further, based on the definition, for every non-parent SCC for example a child SCC~$\mc{S}^c_i$, there exist a parent SCC~$\mc{S}^p_j$ for which~$\mc{S}^c_i \prec \mc{S}^p_j$; this implies that there is a path from all nodes in~$\mc{S}^c_i$ to nodes in~$\mc{S}^p_j$, which are output-connected and output-connectivity of nodes in~$\mc{S}^c_i$ follows. This holds for every non-parent SCC, as it has a path to (at least) one parent SCC and thus is connected to (at least) one output.
	
	\textbf{Necessity}: Assume that there is only one parent SCC~$\mc{S}^p_j$ with no output. Following the definition of a parent SCC, there is no path from nodes in~$\mc{S}^p_j$ to any other measured SCC or any direct measurement, violating condition (i) in Lemma~\ref{lem_cent}.
	
	Since one measurement from one node in every parent SCC is necessary and sufficient, the minimum number of non-zero entries in~$\mc{H}$ is equal to number of parent SCCs~$|\mc{S}^p|$ in the network; and the theorem follows.
\end{proof}
A very important concept is irreducibility of the adjacency matrix defined as follows,

\begin{defn}
    An \textit{irreducible} matrix is such that it cannot be transformed into block upper-triangular or block lower-triangular by simultaneous row/column permutations. A matrix that is not irreducible is \textit{reducible}.
\end{defn}
The concept of strong-connectivity and irreducibility are related according to the following:
\begin{lem} \label{lem_SC}
	For a SC network~$\mc{G}$, its adjacency matrix~$A$ is irreducible.
\end{lem}
\begin{proof}
	The proof is given in~\cite{rein_book}.
\end{proof}
To better understand the structure of an irreducible matrix  we provide an example. For a typical~$n \times n$ irreducible matrix, there exists a sequence of~$n$ non-diagonal nonzero entries that share no row and no column~\cite{rein_book}. In a graph-theoretic perspective, this sequence represents a path that spans all nodes in the associated graph. For example an irreducible matrix with~$4$ components may have the following structure:
	\begin{equation}  \nonumber
	\left(
	\begin{array}{c|c|c|c}
	  & & & A_{14}\\
	\hline
	A_{21} &  & &  \\ \hline
	 & A_{32} & &  \\ \hline
	&  & A_{43} & 	
	\end{array} \right)
	\end{equation}

\begin{lem} \label{lem_self}
	For a self-damped network~$\mc{G}$ with adjacency matrix~$A$, the output/input connectivity condition for observability/controllability recovery is that every parent/child SCC is measured/controlled.
\end{lem}
\begin{proof}
   The key point is that  the self-damped network~$\mc{G}$ contains a disjoint family of self-cycles spanning all nodes, and therefore from Lemma~\ref{lem_srank} its structured adjacency matrix~$\mc{A}$ is full S-rank. Then, following similar procedure as in proof of Theorem~\ref{thm_scc}, measuring (at least) one state in every parent SCC recovers output-connectivity for observability, and controlling (at least) one state in every child SCC recovers input-connectivity for controllability,
\end{proof}

The idea for Lemma~\ref{lem_self} is taken from the proof methodology in chapter 1 in~\cite{rein_book}. For a self-damped network~$\mc{G}$, all diagonal entries of adjacency matrix~$A$ are non-zero. A path from a node~$i$ to a node~$j$ in network~$\mc{G}$ represents a sequence of non-zero entries in its adjacency matrix~$A$ from diagonal entry~$A_{ii}$ to the diagonal entry~$A_{jj}$, as illustrated in Fig.~\ref{fig_matrix}.
\begin{figure}
	\centering
	{\includegraphics[width=3.5in]{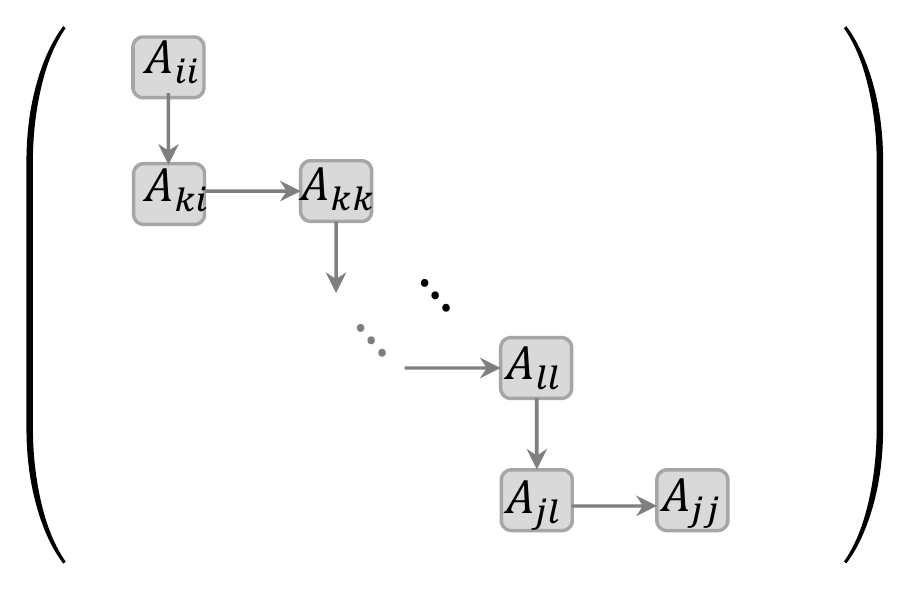}}
	\caption{ This figure shows the procedure to check the output/input-connectivity via the nonzero pattern of the network adjacency matrix. This matrix may represent a self-damped network with non-zero diagonal entries, or block triangularization of adjacency matrix of a non-SC network with irreducible block diagonals. A sequence of nonzero entries/blocks between~$\mc{A}_{ii}$ and~$\mc{A}_{jj}$ as shown in the figure represents a path between two nodes/SCCs~$i$ and~$j$ in the network. The input-connectivity of node/SCC~$i$ implies the input-connectivity of all nodes in the path from node/SCC~$i$ to~$j$, and, similarly, the output-connectivity of node/SCC~$j$ implies that all nodes in the path from node/SCC~$i$ to~$j$ are connected to the  output.}
	\label{fig_matrix}
\end{figure}
Note that the non-zero entries represent the weights of the links in the sequence of nodes from node~$i$ to node~$j$.
One may adopt this approach to easily check the output/input-connectivity of the structured adjacency matrix~$\mc{A}$. Note that, this fact can be generalized for irreducible block diagonals of the adjacency matrix. A non-SC network can be transformed into a block triangular form with irreducible block diagonals, each block representing a SCC, and lower diagonal blocks, defining the partial order of SCCs. For the illustrative example in Fig.~\ref{fig_matrix}, the diagonals~$A_{ii}$,~$A_{kk}$,~$A_{ll}$, and~$A_{jj}$ may represent an irreducible block (a SCC);~$A_{ki}$ and~$A_{jl}$ represent a block with at least one non-zero entry. For this example, the partial order of SCCs is defined as~$\mc{S}^c_i \prec \mc{S}_k \prec \mc{S}_l \prec \mc{S}^p_j$.

\begin{lem}
	All parent and child SCCs are disjoint and share no nodes.
\end{lem}
\begin{proof}
	As proved in~\cite{algorithm} all SCCs obtained from a SCC decomposition (including the parent and child SCCs) are disjoint.
\end{proof}
This lemma holds generally for all SCCs irrespective of their partial order.
Note that, the input/output-connectivity of SCCs are related to condition (i) in Lemma~\ref{lem_cent}. For condition (ii), other than SCCs, graph notions and components regarding the S-rank deficiency of networks are involved in the observability/controllability.

\begin{defn} \label{def_unmatched}
	For a network, define a \textit{maximum matching}~$\mc{M}$ as the maximum size set of links that share no common nodes. A node is \textit{matched} for observability/controllability if it is end/begin node of a link in~$\mc{M}$; otherwise, it is \textit{unmatched}.
\end{defn}

It is known that the unmatched nodes are not unique, in general. The set of all possible unmatched nodes for observability/controllability is represented by new components described as follows.

\begin{defn} \label{def_contraction}
	Define a dilation as a subset of nodes~$\mc{F}$ such that~$|\mc{N}(\mc{F})|<|\mc{F}|$; where the neighborhood~$\mc{N}(\mc{F})$ is the set of all nodes~$\mc{V}_i$,  which are the begin nodes of links to nodes~$\mc{V}_j$ in~$\mc{F}$, i.e.,~$\mc{N}(\mc{F}) = \{\mc{V}_i|(\mc{V}_i,\mc{V}_j) \in \mc{E}, \mc{V}_j \in \mc{F}\}$. Similarly, define a contraction as a subset of nodes~$\mc{F}$ such~that~$|\mc{N}(\mc{F})|<|\mc{F}|$, where~$\mc{N}(\mc{F})=\{\{\mc{V}_j|(\mc{V}_i,\mc{V}_j) \in \mc{E}, \mc{V}_i \in \mc{F}\}$
\end{defn}

Roughly speaking, contractions are involved with the subgraph in which more number of nodes are linked/contracted towards less other nodes. Similarly, in a dilation less number of nodes are linked/dilated to more other nodes. It is known that all nodes in the same dilation/contraction are equivalent for controllability/observability, and therefore, one unmatched node from every contraction/dilation is necessary and sufficient for S-rank recovery~\cite{doostmohammadian2017recovery,doostmohammadian2017observational}. In this direction, define the set~$\mc{C}=\{\mc{C}_1,\mc{C}_2,\dots\}$ and~$\mc{D}=\{\mc{D}_1,\mc{D}_2,\dots\}$ as the set of all contractions and dilations, respectively.
In general, the sets~$\mc{D}$ and~$\mc{C}$ are unique and, further,  the contractions/dilations are not disjoint  and may share nodes~\cite{doostmohammadian2017observational}.

\begin{thm} \label{thm_unmatched}
	For observability/controllability of SC network~$\mc{G}_2$, it is necessary and sufficient to measure/control every unmatched node in the network contractions/dilations, i.e., for observability~$|\mc{H}|_0=|\mc{C}|$ and for controllability~$|\mc{H}|_0=|\mc{D}|$.
\end{thm}
\begin{proof}
	The proof is given in our previous work~\cite{jstsp14,doostmohammadian2017observational} for observability.
In fact, for any choice of maximum matching, there is one unmatched node in every dilation and similar statement holds for contractions. This directly follows the definition of contraction/dilation. Further, in case two dilations/contractions share a node, controlling/measuring the shared node recovers the S-rank by one~\cite{doostmohammadian2017observational,globalsip14}. Therefore, for full S-rank recovery one node in every dilation/contraction must be  controlled/measured. 	
\end{proof}
From Theorem~\ref{thm_unmatched} and Lemma~\ref{lem_srank2} one can conclude that having an observation/input at every unmatched node recovers the S-rank; in terms of observability:
\begin{eqnarray} \nonumber
	\mbox{S-rank}\left(
	\begin{array}{c}
	\mc{A}_2\\
	\mc{H}
	\end{array}
	\right) = n.
\end{eqnarray}
and in terms of controllability:
\begin{eqnarray} \nonumber
	\mbox{S-rank}\left(\mc{A}_2|\mc{H}	\right) = n.
\end{eqnarray}
Although in Theorem~\ref{thm_scc} the results are given for~$\mbox{S-rank}(\mc{A})=n$, one can easily restate  Theorem~\ref{thm_scc} based on the above equations. Similar statement holds for Theorem~\ref{thm_unmatched}. In other words, having the observations/inputs according to Theorem~\ref{thm_scc} for output/input-connectivity (instead of having SC network), the S-rank can be recovered according to Theorem~\ref{thm_unmatched}. Similarly,
having the observations/inputs according to Theorem~\ref{thm_unmatched} for S-rank recovery (instead of having full S-rank network), the input/output-connectivity can be recovered via Theorem~\ref{thm_scc}. However, some inputs/observations may recover both conditions. Combining Theorems\ref{thm_scc} and~\ref{thm_unmatched}, one can  deduce the following corollary.
\begin{cor} \label{cor_minH}
In general, contractions/dilations may share nodes with SCCs. Therefore, from Theorems~\ref{thm_scc} and~\ref{thm_unmatched}, the minimum number of non-zeros in~$\mc{H}$ for observability is,
\begin{eqnarray}
|\mc{H}|_0 = |\mc{S}^p|+|\mc{C}|-\min(|\mc{S}^p \cap \mc{C}|)
\end{eqnarray}
and for controllability,
\begin{eqnarray}
|\mc{H}|_0 = |\mc{S}^c|+|\mc{D}|-\min(|\mc{S}^c \cap \mc{D}|)
\end{eqnarray}
where~$|\mc{S}^p \cap \mc{C}|$ (respectively~$|\mc{S}^c \cap \mc{D}|$) represents the number of parent SCCs and contractions (respectively child SCCs and dilations) that share nodes.
\end{cor}

The above corollary generalizes the results in previous work by the authors~\cite{asilomar11, jstsp14}.
In~\cite{asilomar11}, we introduced the concepts of parent SCCs for observability recovery and output-connectivity in networked estimation, while in~\cite{jstsp14} the concepts of contractions and maximum matching are adopted for S-rank recovery. Similar results are stated in~\cite{pequito2015framework}, where the concepts of \textit{non-bottom/top linked SCCs}, \textit{left/right-unmatched nodes}, and \textit{maximum bottom/top assignability index} are introduced. As compared to~\cite{pequito2015framework}, Theorems~\ref{thm_scc} and~\ref{thm_unmatched} and Corollary~\ref{cor_minH}, using dilations/contractions which contain all equivalent nodes in terms of controllability/observability, provide \textit{all} set of nodes satisfying minimal controllability/observability, while~\cite{pequito2015framework}, using right/left-unmatched nodes, gives only \textit{one} minimal set. Note that state nodes in contractions/dilations are all equivalent in terms of observability/controllability. This is particularly of interest, as using the results of our work, one can choose the \textit{minimal-cost} set of measurements/inputs among all possible options. 
Recall that, measuring/controlling the shared node between a parent/child SCC and a contraction/dilation recovers both conditions in Lemma~\ref{lem_cent}. We refer interested readers to~\cite{icassp13} for examples of S-rank-deficient networks with specific structures  containing contractions/dilations.

\section{Main Results} \label{sec_main}
In the previous section, we derive the minimal conditions  to solve Problem~\ref{probA}. These conditions are required to find the solution for Problem~\ref{prob}. Assuming the conditions in Theorem~\ref{thm_scc} and Theorem~\ref{thm_unmatched} (or Corollary~\ref{cor_minH}) are satisfied, we extend the results to sufficient conditions to solve Problem~\ref{prob}.

\begin{thm} \label{thm_main1}
	Let~$\mbox{S-rank}(\mc{A}_2)<n$ and~$\mc{H}$ includes the minimum  measurements/control-inputs for S-rank recovery of ~$\mc{A}_2$ (according to Theorem~\ref{thm_unmatched}). Then,~$|\mc{H}_C|_0$ for  S-rank recovery of~$(\mc{A}_1 \otimes \mc{A}_2)$ is minimized if~$\mbox{S-rank}(\mc{A}_1)=N$. Further, for a self-damped network~$\mc{G}_1$, a sufficient condition on the matrix~$\mc{H}_C$ to recover the S-rank is that~$\mc{H}_C = I_N \otimes \mc{H}$, where~$I_N$ is the~$N$ by~$N$ identity matrix.
\end{thm}
\begin{proof}
	 Let~$\mbox{S-rank}(\mc{A}_2)=m<n$. Based on Lemma~\ref{lem_srankproduct},
	\begin{equation} \nonumber
	\mbox{S-rank}(\mc{A}_1 \otimes \mc{A}_2) \geq \mbox{S-rank}(\mc{A}_1)\cdot m.
	\end{equation}	
	Following Lemma~\ref{lem_srank2} for S-rank recovery,~$\mc{H}_C$ must be designed such that
	\begin{equation} \nonumber
	\mbox{S-rank}\left(
	\begin{array}{c}
	\mc{A}_1 \otimes \mc{A}_2\\
	\mc{H}_C
	\end{array}
	\right) = Nn.
	\end{equation}
	Noting that every non-zero entry in~$\mc{H}_C$  may recover (at most) one S-rank deficiency of~$\mc{A}_1 \otimes \mc{A}_2$ and from Lemma~\ref{lem_srankproduct}, we have
	\begin{eqnarray} \nonumber
	\min(|\mc{H}_C|_0) &=&\min(Nn-\mbox{S-rank}(\mc{A}_1 \otimes \mc{A}_2)) \\ \nonumber
	&=& Nn-\max(\mbox{S-rank}(\mc{A}_1 \otimes \mc{A}_2))\\ \nonumber
	&\leq& Nn-\max(\mbox{S-rank}(\mc{A}_1) \mbox{S-rank}(\mc{A}_2)) \\ \nonumber
	&\leq& Nn-\max(\mbox{S-rank}(\mc{A}_1))m.
	\end{eqnarray}
	Therefore, to minimize~$|\mc{H}_C|_0$ we need~$\mbox{S-rank}(\mc{A}_1)=N$, which results in
	\begin{eqnarray} \label{eq_H_c}
	\min(|\mc{H}_C|_0) \leq N(n-m).
	\end{eqnarray}	
	From Lemma~\ref{lem_srankproduct} the equality in equation~\eqref{eq_H_c} holds for almost all choices of ~$A_1$ and~$A_2$, i.e., for almost all cases of network products we have~$\min(|\mc{H}_C|_0) = N(n-m)$.
		
	To prove the second part, we borrow the ideas from the structural methodology in~\cite{rein_book}. From Lemma~\ref{lem_srank2}, we have
	\begin{equation}  \nonumber
	\mbox{S-rank}\left(
	\begin{array}{c}
	\mc{A}_2\\
	\mc{H}
	\end{array}
	\right) = n.
	\end{equation}	
	Let~$A_1=[\mathbf{a}_1, \mathbf{a}_2,\dots, \mathbf{a}_N]$, where~$[\mathbf{a}_i]$ represents a column of~$A_1$. From Lemma~\ref{lem_srank}, network~$\mc{G}_1$ being self-damped implies that ~$\mbox{S-rank}[\mathbf{a}_1, \mathbf{a}_2,\dots, \mathbf{a}_N]=N$ and ~$\mbox{S-rank}[\mathbf{a}_i]=1, ~i=\{1,\dots,N\}$. Therefore, from the definition of S-rank and from Lemma~\ref{lem_srankproduct}, we have 	
	\begin{equation}
	\mbox{S-rank}\left(
	\begin{array}{c}
	[\mathbf{a}_i] \otimes \mc{A}_2\\
	\mc{H}
	\end{array}
	\right) = n.
	\end{equation}
	From another perspective, since the~$i$th entry of column~$[\mathbf{a}_i]$ is non-zero and following the definition of the Kronecker product, the~$i$th~$n \times n$ block of the~$nN \times n$ matrix~$[\mathbf{a}_i] \otimes \mc{A}_2$ represents the matrix~$\mc{A}_2$. We note that~$\mc{H}$ recovers the S-rank of~$\mc{A}_2$; this is a sufficient condition for S-rank recovery of~$[\mathbf{a}_i] \otimes \mc{A}_2$.
	
	Based on the Kronecker product definition, the structure of~$\mc{A}_1 \otimes \mc{A}_2$ is obtained as a side by side concatenation of matrices~$[\mathbf{a}_i] \otimes \mc{A}_2,~i=\{1,\dots,N\}$, where~$\mbox{S-rank}[\mathbf{a}_1, \mathbf{a}_2,\dots, \mathbf{a}_N]=N$. Having~$\mc{G}_1$ as a self-damped network implies that the~$i$th diagonal block of~$[\mathbf{a}_1, \mathbf{a}_2,\dots, \mathbf{a}_N] \otimes \mc{A}_2$ is~$\mc{A}_2$.
	Therefore, the~$i$th block of~$I_N \otimes \mc{H}$ recovers the S-rank of~$i$th block of~$\mc{A}_1 \otimes \mc{A}_2$ as a sufficient condition, i.e.,
	\begin{equation}  \nonumber
	\mbox{S-rank}\left(
	\begin{array}{c|c|c|c}
	[\mathbf{a}_1] \otimes \mc{A}_2 & [\mathbf{a}_2] \otimes \mc{A}_2& \dots  & [\mathbf{a}_N] \otimes \mc{A}_2\\
	\hline
	\mc{H} & 0 & & 0 \\
	0 & \mc{H} & & \vdots \\
	\vdots & 0 & \ddots & \vdots \\
	\vdots & \vdots & & 0\\
	0 & 0& & \mc{H}	
	\end{array}
	\right) = Nn,
	\end{equation}
where the above can be compactly written as
	\begin{equation} \nonumber
	\mbox{S-rank}\left(
	\begin{array}{c}
	\mc{A}_1 \otimes \mc{A}_2\\
	I_N \otimes \mc{H}
	\end{array}
	\right) = Nn.
	\end{equation}
\end{proof}
Theorem~\ref{thm_main1} implies that given a rank-deficient system~$A_2$ (associated with the replica network~$\mc{G}_2$) when the~$\mbox{S-rank}(A_1)<N$ (i.e., the factor network~$\mc{G}_1$ is rank-deficient) more number of observations of the composite network are required. We further remind the reader  that~$\mc{H}_C=I_N \otimes \mc{H}$ is a sufficient condition for observability of the composite network~$\mc{G}_1 \times \mc{G}_2$.

\begin{thm} \label{thm_main2}
	Let~$\mbox{S-rank}(\mc{A}_2)=n$ and~$\mc{H}$~includes~the~minimum  measurements/control-inputs for output/input-connectivity of~$\mc{G}_2$ (according to Theorem~\ref{thm_scc}). To minimize ~$|\mc{H}_C|_0$ for output/input-connectivity of~$(\mc{G}_1 \times \mc{G}_2)$, it is sufficient that~$\mc{G}_1$ is SC and self-damped. Then,~$|\mc{H}_C|_0=|\mc{H}|_0$.
\end{thm}
\begin{proof}
	The structural approach for the proof is based on the ideas from the methodology in chapter~$1$ of~\cite{rein_book}. We prove the case for observability and the dual case of controllability similarly follows. The minimal number of measurements for observability of composite network~$\mc{G}_1 \times \mc{G}_2$ is equal to the number of measurements for observability of~$\mc{G}_2$ according  to Theorem~\ref{thm_scc}, i.e.,~$\min(|\mc{H}_C|_0)=|\mc{H}|_0$. It is sufficient to prove that for self-damped and SC network~$\mc{G}_1$, we have~$|\mc{H}_C|_0=|\mc{H}|_0$. Based on the results of Corollary~\ref{cor_minH}, the matrix~$\mc{A}_2$ being full-rank implies that~$|\mc{H}|_0=|\mc{S}^p|$, and therefore, it is sufficient to prove that~$|\mc{H}_C|_0 = |\mc{S}^p|$.
	
	Since~$\mc{G}_1$ is self-damped, the non-vanishing main diagonal blocks in the adjacency matrix of Kronecker composite network,~$\mc{A}_1 \otimes \mc{A}_2$, are reducible structured matrices having the structure of~$\mc{A}_2$. This is because, based on the Kronecker product definition, multiplying  each scalar entry in~$A_2$ does not change its structure~$\mc{A}_2$. Assume in~$\mc{G}_2$ a child SCC~$\mc{S}^c_i$  connected via a path to a parent SCC~$\mc{S}^p_j$, i.e.,~$\mc{S}^c_i \prec \mc{S}^p_j$. Then, according to Fig.~\ref{fig_matrix}, in the structured matrix~$\mc{A}_2$ there is a sequence of non-zero blocks from~$i$th irreducible diagonal block to~$j$th irreducible diagonal block.
	
	Further, the non-diagonal blocks of~$\mc{A}_1 \otimes \mc{A}_2$ are mapped based on the irreducible structure of~$\mc{A}_1$, see Fig.~\ref{fig_matrix2}.
	\begin{figure}
		\centering
		{\includegraphics[width=3.5in]{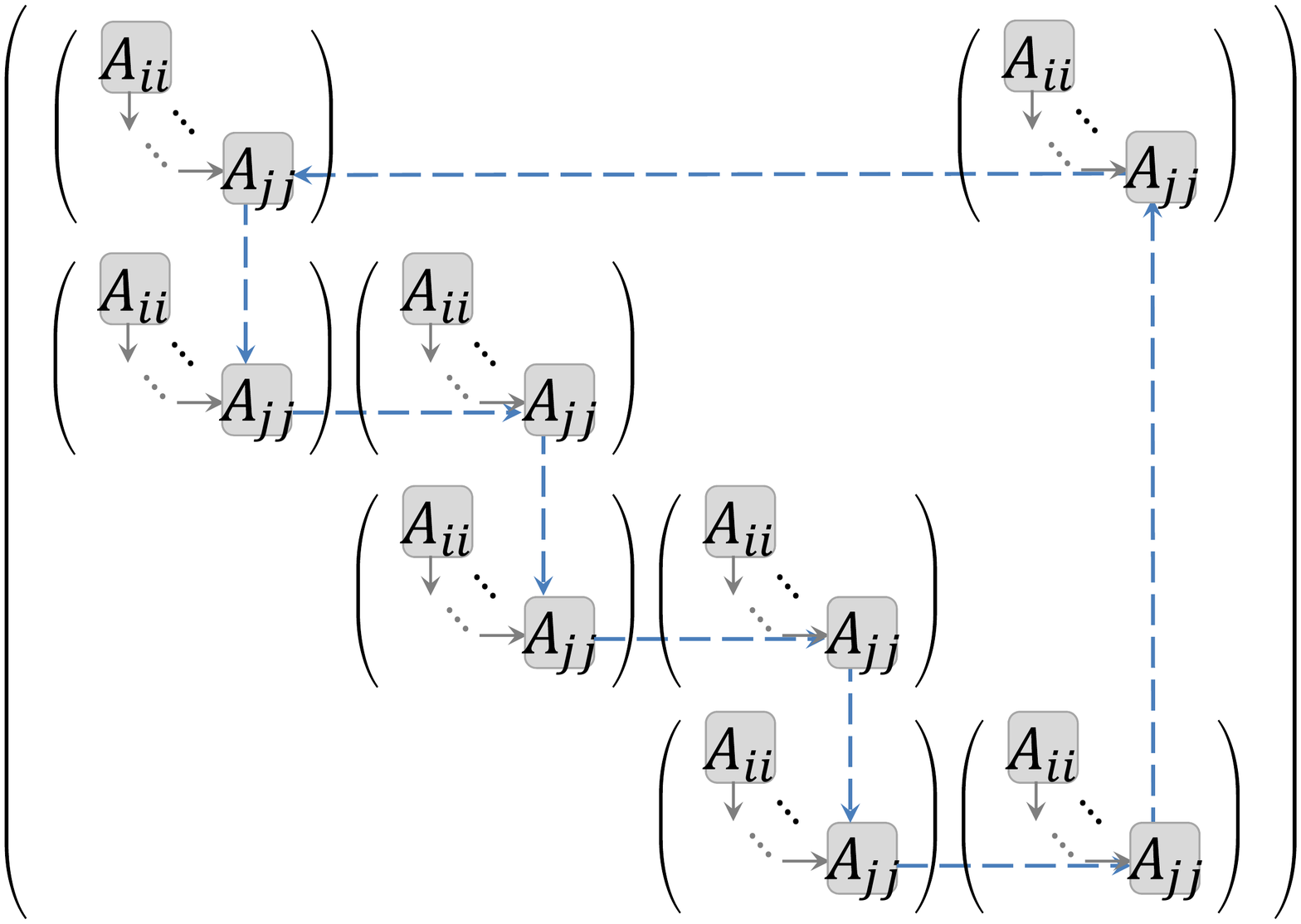}}
		\caption{ This figure shows the structure of~$\mc{A}_1 \otimes \mc{A}_2$. Based on the definition of Kronecker product, each block represents the structure of~$\mc{A}_2$, including an irreducible sub-matrix~$A_{ii}$ (a child SCC in~$\mc{G}_2$) connected to an irreducible sub-matrix~$A_{jj}$ (a parent SCC in~$\mc{G}_2$). The blocks are mapped based on the irreducible and self-damped structure of~$\mc{A}_1$. The irreducible sub-matrices in blocks of~$\mc{A}_2$ make a larger irreducible block via the path shown by the dashed arrows, thus, making a larger irreducible block (by proper row and column permutations) which represents a larger parent SCC in the Kronecker composite graph. }
		\label{fig_matrix2}
	\end{figure}
	Following the results of Lemma~\ref{lem_SC}, there exists a sequence of non-diagonal blocks in~$\mc{A}_1 \otimes \mc{A}_2$, each having the structure of~$\mc{A}_2$, sharing no hyper-row and no hyper-column\footnote{By a hyper-row we refer to a~$n$ by~$nN$ row matrix made by block matrices~$\mc{A}_2$ in the Kronecker matrix~$\mc{A}_1 \otimes \mc{A}_2$. The definition of hyper-column similarly follows.}. As illustrated in Fig.~\ref{fig_matrix2}, the irreducible sub-blocks of parent SCC~$\mc{S}^p_j$ in each~$\mc{A}_2$ block makes a strongly-connected path with each-other via the irreducible structure of~$\mc{A}_1$. This implies that by proper permutation of rows and columns of~$\mc{A}_1 \otimes \mc{A}_2$, one can find a larger irreducible block representing a larger SCC in~$\mc{G}_1 \times \mc{G}_2$. Note that this SCC has no outgoing links and therefore represents a parent SCC in the composite graph~$\mc{G}_1 \times \mc{G}_2$. This generally holds for any type of SCCs in~$\mc{G}_2$, i.e., having a self-damped and SC network~$\mc{G}_1$, for any SCC in the network~$\mc{G}_2$, there exists a larger SCC in the composite network~$\mc{G}_1 \times \mc{G}_2$, while the partial order follows the partial order of SCCs in~$\mc{G}_2$. Therefore, according to Theorem~\ref{thm_scc}, the number of sufficient measurements for observability of the composite network is equal to the number of parent SCCs in~$\mc{G}_1 \times \mc{G}_2$, which is equal to the number of parent SCCs in~$\mc{G}_2$, i.e.,~$|\mc{H}_C|_0 = |\mc{S}^p|$.
\end{proof}

It should be noted that although Theorem~\ref{thm_main2} is stated for full S-rank networks, one can easily restate and prove the theorem for S-rank-deficient networks possessing all the measurements/control-inputs for S-rank recovery according to Theorem~\ref{thm_unmatched}. In other words, Theorem~\ref{thm_main2} holds for either a full S-rank network or a S-rank-deficient network with measurements/inputs satisfying conditions in Theorem~\ref{thm_unmatched}. In such case, one can reduce~$|\mc{H}_C|_0$ by measuring/controlling the shared node between a parent/child SCC and a contraction/dilation.
It should be noted that although Theorems~\ref{thm_main1} and~\ref{thm_main2} consider the self-damped condition for  factor network~$\mc{G}_1$, this is a typical assumption in distributed estimation literature, e.g., see~\cite{asilomar11,usman_cdc:10,park2017design}. This assumption implies that every sensor/estimator applies its own information to develop the state estimation of the dynamical system.

\section{Application in Distributed Estimation} \label{sec_app}
In distributed estimation a network of sensors/agents, represented by network~$\mc{G}_1$, are tasked to monitor a dynamical system, represented by network~$\mc{G}_2$ based on the partial measurements observed by agents.
The observability results in this paper finds direct application in \textit{single time-scale} distributed estimation.
Such distributed estimation protocols, as compared to \textit{multi time-scale} distributed estimation as in Kalman consensus filters~\cite{olfati2:05}, have the advantages of low communication on sensors and no constraints on local observability of agents. In this scenario, the structure of the underlying dynamical system  dictates the structure of sensor network~\cite{doostmohammadian2017recovery}. The proposed protocol in our previous works~\cite{jstsp,jstsp14} is an example application of composite Kronecker network product. The protocol has two steps of \textit{prediction fusion} and \textit{observation fusion} as follows:
\begin{eqnarray}\label{eq_estimation1}
\widehat{\mb{x}}^i_{k|k-1} &=& \sum_{j\in\mathcal{N}(i)} W_{ij}A\widehat{\mb{x}}^j_{k-1|k-1}, \\ \label{eq_estimation2}
\widehat{\mb{x}}^i_{k|k} &=&\widehat{\mb{x}}^i_{k|k-1} + K_k^i \sum_{j\in \mc{N}(i)}H_j^\top \left(\mb{y}^j_k-H_j\widehat{\mb{x}}^i_{k|k-1}\right).
\end{eqnarray}
where~$W = [W_{ij}]$ is the adjacency  matrix of~$\mc{G}_1~$,~$A$ is the adjacency  matrix of~$\mc{G}_2~$,~$H_j$ is the measurement matrix at agent~$j$,~$\widehat{\mb{x}}^i_{k|k-1}$ is the state prediction at agent~$i$ given all the neighboring  measurements up to time~$k-1$, and~$\widehat{\mb{x}}^i_{k|k}$ is the state estimate at agent~$i$ given all the neighboring measurements up to time~$k$. The block-diagonal matrix~${K}_k=\mbox{blockdiag}[K_k^1,\ldots,K_k^N]$ represents the gain matrix. The error dynamics for the protocols~\eqref{eq_estimation1}-\eqref{eq_estimation2} is as follows~\cite{ISJ}:
\begin{eqnarray}\label{eq_err1}
\mb{e}_{k} = (W\otimes A - K_kD_H(W\otimes A))\mb{e}_{k-1} +
\mb{q}_k,
\end{eqnarray}
where~$\mb{q}_k$ collects the noise terms,~$\mb{e}_{k}^i = \mb{x}_{k|k} - \widehat{\mb{x}}^i_{k|k}$ is the estimation error at agent~$i$,~$\mb{e}_{k}$ is the global estimation error defined as,
\begin{eqnarray}\nonumber
\mb{e}_{k} &=& \left( \begin{array}{c}
\mb{e}_{k}^1\\
\vdots \\
\mb{e}_{k}^N
\end{array}\right)
\end{eqnarray}
and~$D_H$ represents the global measurement matrix associated with distributed estimation.
It is well known that the error dynamics~\eqref{eq_err1} are bounded if the pair~$(W \otimes A, D_H)$ is observable, also known as \textit{distributed observability}~\cite{jstsp,jstsp14}. The distributed observability, in fact, refers to the observability of the composite network~$\mc{G}_1 \times \mc{G}_2$, where~$\mc{G}_1$ and~$\mc{G}_2$ are respectively associated to matrices sensor communication~$W$ and the system matrix~$A$.

Following Theorem~\ref{thm_main1}, if the dynamical system~$A$ is S-rank-deficient, a full S-rank sensor network minimizes the required number of measurements for distributed observability. Further, if the sensor network is self-damped then for distributed observability it is sufficient that every sensor possesses (in its neighborhood) all the measurements required  for observability of the dynamical system. This assumption, for example, is given in distributed estimation protocol proposed by~\cite{das2015distributed}. Note that the self-link at each agent implies that each agent uses its own state predictions in equation~\eqref{eq_estimation1}, which is a typical assumption in distributed estimation scenarios. Following Theorem~\ref{thm_main2}, if the dynamical system~$A$ is full S-rank, an SC sensor network  ensures distributed observability. Similar result is used  in distributed estimation literature, as in~\cite{sayedtu12}. Further, having an SC self-damped sensor network, Theorem~\ref{thm_main2} implies that the minimum number of measurements ensuring observability of dynamical system also ensures distributed observability of the distributed estimator~\eqref{eq_estimation1}-\eqref{eq_estimation2}. 

In the same line of research,~\cite{zhou2013coordinated} proposes a distributed estimator to predict the state of networked dynamical system. Sensors take local measurements of the system and share their predictions over a network with exact \textit{same} structure as the networked  system. This implies that in the composite network scenario, the structure of the factor network~$\mc{G}_1$ and the replica network~$\mc{G}_2$ must be the same. The results of~\cite{zhou2013coordinated} are further extended in~\cite{zhou2015controllability} where a distributed algorithm for observability analysis of networked system is proposed based on PBH test. Also, some necessary and sufficient conditions on system matrix are given that guarantee the performance of the distributed estimator is equivalent with the lumped Kalman filter. As a comparison, in the proposed distributed estimator~\eqref{eq_estimation1}-\eqref{eq_estimation2}, the sensor network~$\mc{G}_1$ and the dynamical system~$\mc{G}_2$ have different structures, as we show that the structure of~$\mc{G}_1$ depends on the S-rank of ~$\mc{G}_2$. For example, if~$\mc{G}_2$ is full S-rank, then a \textit{sparsely-connected} SC network~$\mc{G}_1$ is sufficient for distributed estimation, while the estimation performance could be improved by adopting more densely-connected sensor networks. In~\cite{jstsp,icassp13}, we clearly compare the performance of  distributed estimator~\eqref{eq_estimation1}-\eqref{eq_estimation2} with centralized Kalman filter. We again recall that here our main goal is the minimal sufficient conditions to ensure observability, while estimation performance is not the focus of this paper and left for future research direction. Finally, it should be mentioned that the results of this paper are not restricted to distributed estimation, but also find application in, for example, networked control systems and other scenarios mentioned in Section~\ref{sec_intro}.

\section{Illustrative Example and Simulation} \label{sec_sim}
In this section, we consider a composite graph example and investigate its observability based on the results of previous sections. Consider the replica network~$\mc{G}_2$ as shown in Fig.~\ref{fig_example}.
\begin{figure}
	\centering
	{\includegraphics[width=3.5in]{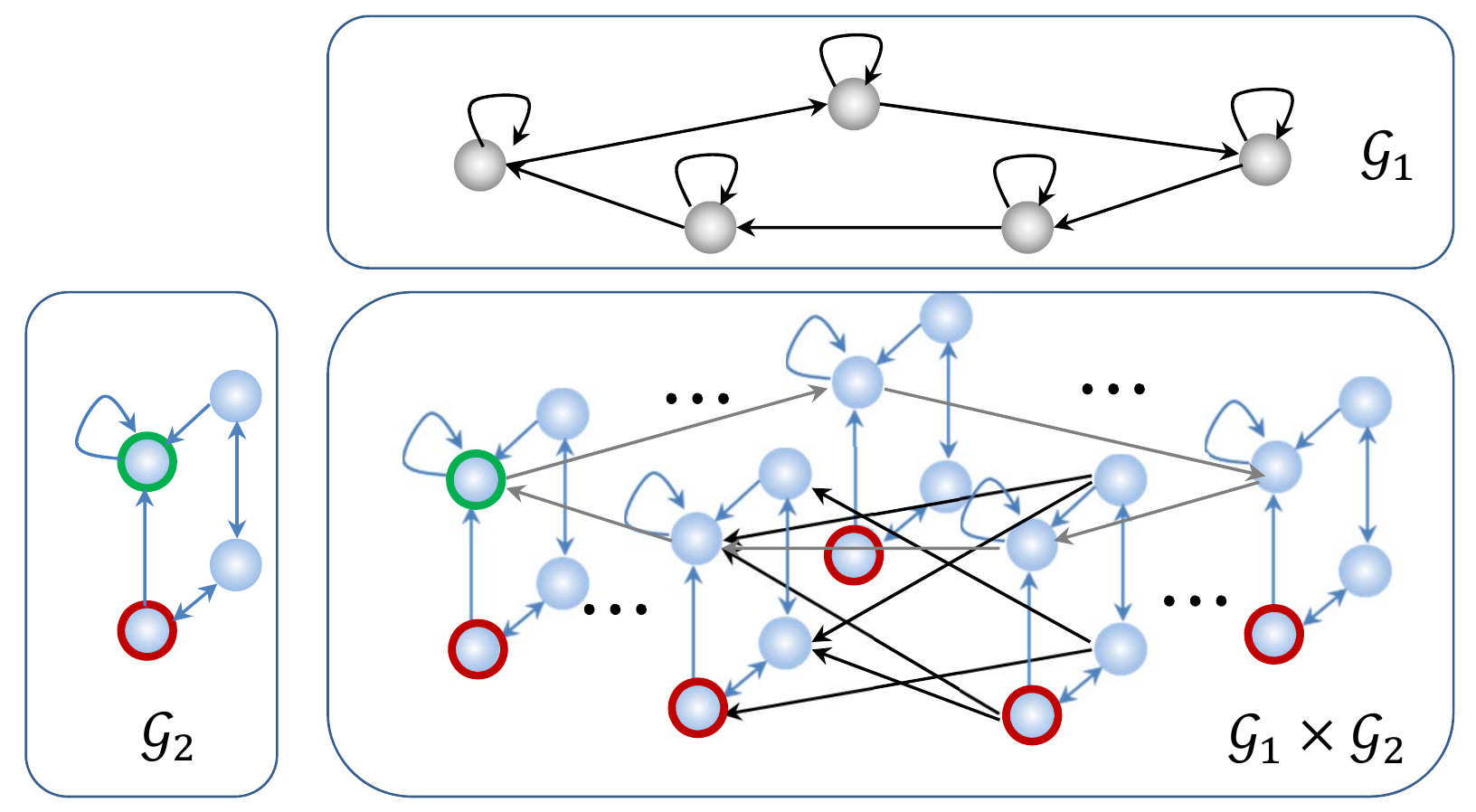}}
	\caption{ This figure shows an example of Kronecker product~$\mc{G}_1 \times \mc{G}_2$ of two graphs. The factor network~$\mc{G}_1$ is designed such that it is SC and self-damped. The links among two replica networks are shown and the rest of the links are skipped for clarity of the figure. Only the links involved in the main parent SCC of the composite network are shown in gray color. The red-border nodes are observed as sufficient condition for S-rank recovery (Theorem~\ref{thm_main1}), and the green-border node is observed as sufficient condition for output-connectivity (Theorem~\ref{thm_main2}).}
	\label{fig_example}
\end{figure}
This graph contains two SCCs, among which the self-cycle has no outgoing link and therefore is a parent SCC. The S-rank deficiency of the network~$\mc{G}_2$ is~$1$ implying existence of one unmatched node. We find the unmatched node using the \textit{Dulmage-Mendelsohn} algorithm~\cite{dulmage58}. For observability, the green-border node in the network~$\mc{G}_2$ is observed as a parent SCC (Theorem~\ref{thm_scc}) and the red-border node is observed as the unmatched node (Theorem~\ref{thm_unmatched}). The goal is to design the factor network~$\mc{G}_1$ to minimize the number of sufficient measurements for the observability of the Kronecker composite network~$\mc{G}_1 \times \mc{G}_2$. According to Theorem~\ref{thm_main1}, for minimal S-rank recovery, it is sufficient that~$\mc{G}_1$ is full S-rank. We consider the network~$\mc{G}_1$ to be self-damped. Then we have~$\mc{H}_C = I_N \otimes \mc{H}$, implying the observation of the unmatched node in every replica network in~$\mc{G}_1 \times \mc{G}_2$, as shown by red-border nodes in Fig.~\ref{fig_example}. According to Theorem~\ref{thm_main2}, for output-connectivity, it is sufficient that the factor network~$\mc{G}_1$ is SC. Following the results of Theorem~\ref{thm_main2}, there is one parent SCC in the composite network~$\mc{G}_1 \times \mc{G}_2$ as shown in Fig.~\ref{fig_example}, and the observed node is shown by green-border.

The composite network may represent a sensor network monitoring a dynamical system as discussed in Section~\ref{sec_app}. The distributed estimation performance under the protocols~\eqref{eq_estimation1}-\eqref{eq_estimation2} is compared with Kalman filter in our previous works~\cite{jstsp,icassp13}. Here, we apply centralized Kalman filter on  the composite network representing  a linear system as follows:
\begin{eqnarray}
x(k+1) = (A_1 \otimes A_2) x(k) + v(k).
\end{eqnarray}
We consider initial state ~$x(0)$ to be random. The entries in~$A_1 \otimes A_2$ (link weights in~$\mc{G}_1 \times \mc{G}_2$) are considered randomly such that the composite system is potentially unstable with~$\rho=1.93$, where~$\rho$ is the spectral radius of the system matrix. Note that the random link weight consideration is a result of the generic approach adopted in this work, which implies that the observability results hold for almost all values of system parameters. We consider the system noise~$v(k)$ and measurement noise as~$\mc{N}(0,0.05)$ and perform  simulation over~$1000$ Monte-Carlo trials. The time evolution of the Mean Squared Estimation Error (MSEE) using the~$6$ node measurements  is shown in Fig.~\ref{fig_sim}. Note that although the composite system is potentially unstable, MSEE is bounded steady-state stable implying that the composite system is observable.
\begin{figure}
	\centering
	{\includegraphics[width=3.5in]{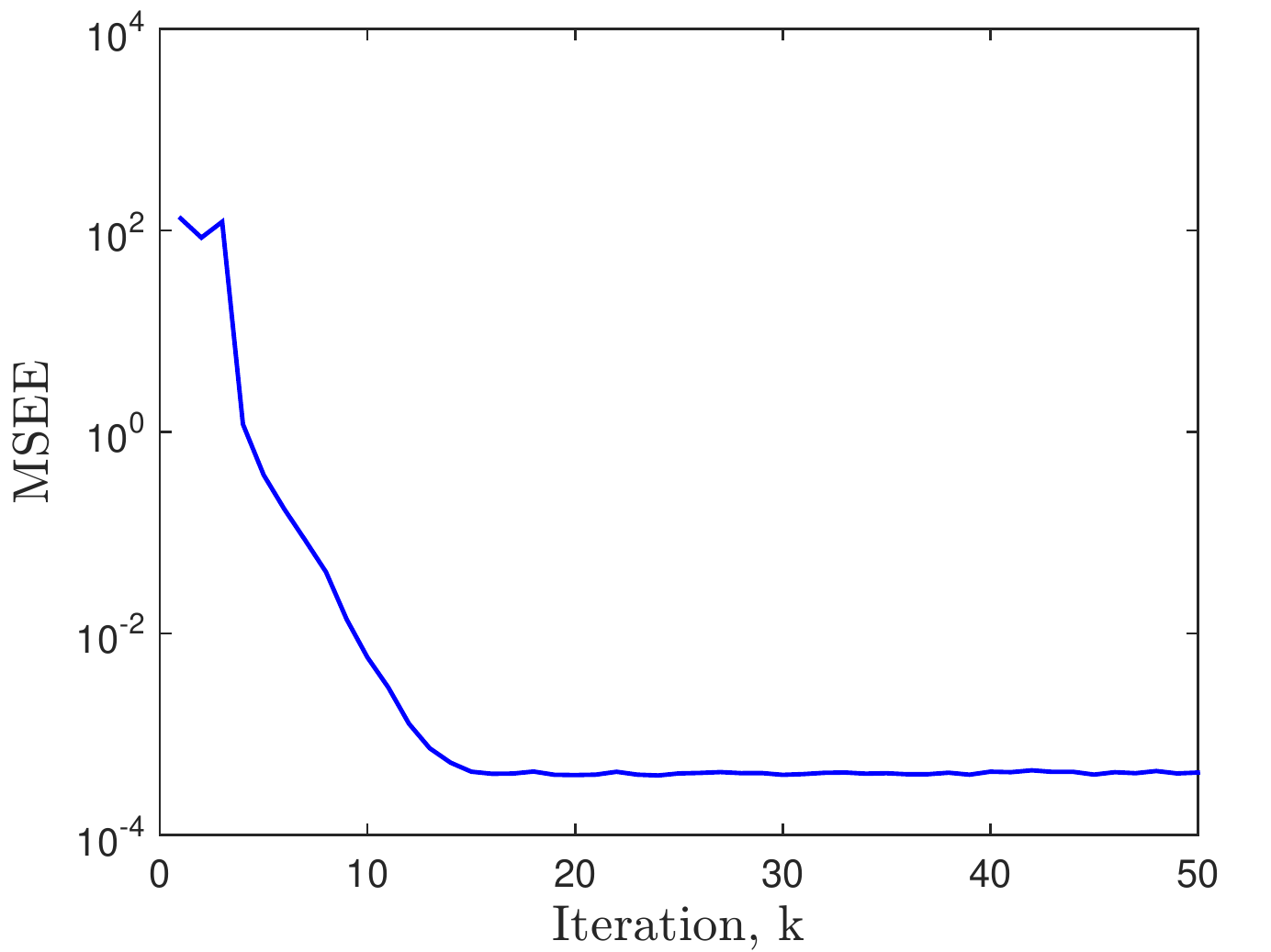}}
	\caption{ This figure shows the time evolution of MSEE of the Kalman filtering over the composite network in Fig.~\ref{fig_example}. }
	\label{fig_sim}
\end{figure}

\section{Concluding Remarks} \label{sec_con}
In this paper, we study the minimal sufficient conditions for observability and controllability of the Kronecker composite networks based on the S-rank and strong-connectivity of its constituent networks. Particularly, the results in Section~\ref{sec_main} are mainly discussed for full S-rank and self-damped constituent networks, which are common assumptions, for example, in distributed estimation. As a general comment, note that controllability and observability are dual concepts and all the results in this paper can be extended to other case by transposing the adjacency matrix and reversing the link directions in the network. We note that the algorithm to find the S-rank of the network is based on Dulmage-Mendelsohn decomposition~\cite{dulmage58} with computational complexity~$\mc{O}(n^{2.5})$. The well-known algorithm for classification of SCCs and defining their partial order is \textit{Depth-First-Search}~\cite{algorithm} with complexity~$\mc{O}(n^2)$. The polynomial-order algorithms support the application in large-scale systems. As the direction of future research, optimal design of large-scale sensor networks for distributed estimation and large-scale networked control systems is a promising field of research.

\bibliographystyle{IEEEbib}
\bibliography{bibliography}

\end{document}